\documentclass[a4paper,UKenglish]{lipics-v2016-arXiv}
\usepackage{graphicx}
\usepackage{algorithm}
\usepackage{algpseudocode}
\usepackage{multicol}
\usepackage{amsmath}
\usepackage{amsthm}
\usepackage{amssymb}
\usepackage{float}
\usepackage{cleveref}[2012/02/15]
\usepackage{relsize}

\usepackage{microtype}

\newcommand{\lastIDX}{\ensuremath{k\cdot n + m}}
\newcommand{\newInputLetter}{\ensuremath{{\rho}}}
\newcommand{\newInput}{\ensuremath{\bar{\newInputLetter}}}
\newcommand{\UBnItemsPerBlock}{\max\set{\floor{\mu},1}}
\newcommand{\UBnBlocks}{\floor{n/\nu}}

\newcommand{\bSize}{\ceil{\mu}}
\newcommand{\bContentSize}{\max\set{\floor{\mu^{-1}},1}}
\newcommand{\nBlocks}{\floor{n/\bSize}}
\newcommand{\lnrLBVal}{\nBlocks \log\big({\bContentSize + 1}\big)}
\newcommand{\lnrLB}{\lnrLBVal = \floor{\frac{n}{\ceil{\lnrErrSymbol / \ell}}} \log\big({\max\set{\floor{\ell / \lnrErrSymbol},1} + 1}\big)}
\newcommand{\lnrLBSymbol}{\mathcal B_{\ell,n,\lnrErrSymbol}}

\newcommand{\minDelta}{1}
\newcommand{\qsr}[1][\minDelta]{\ensuremath{(\ell,n,#1)}-\emph{Sliding Ranker}}
\newcommand{\lnrErrSymbol}{\Delta}
\newcommand{\sensitivity}{\ensuremath{\widetilde{\lnrErrSymbol}}}
\newcommand{\lnr}[1][\lnrErrSymbol]{\ensuremath{(\ell,n,#1)}-\emph{Ranker}}

\newcommand{\total}{\ensuremath{\mathit{total}}}
\newcommand{\subTotal}{\ensuremath{\mathit{subTotal}}}

\newcommand{\ceil}[1]{ \left\lceil{#1}\right\rceil}
\newcommand{\floor}[1]{ \left\lfloor{#1}\right\rfloor}

\newcommand{\parentheses}[1]{ \left({#1}\right)}

\newcommand{\biggParentheses}[1]{ \bigg({#1}\bigg)}

\newcommand{\logp}[1]{\log\parentheses{#1}}
\newcommand{\Omegap}[1]{\Omega\parentheses{#1}}
\newcommand{\Op}[1]{O\parentheses{#1}}
\newcommand{\Thetap}[1]{\Theta\parentheses{#1}}

\newcommand{\cdotpa}[1]{\cdot\parentheses{#1}}

\newcommand{\oneOverE}{ \eps^{-1} }
\newcommand{\oneOverT}{ \tau^{-1} }
\newcommand{\range}[2][0]{#1,1,\ldots,#2}
\newcommand{\orange}[1]{\set{1,2,\ldots,#1}}
\newcommand{\frange}[1]{\set{\range{#1}}}

\newcommand{\smallMultError}{(1+o(1))}
\newcommand{\sFactor}{(1+o(1))}
\newcommand{\sNegFactor}{(1-o(1))}
\newcommand{\brackets}[1]{\left[#1\right]}

\newcommand{\nBits}{\mathfrak b}
\newcommand{\ranker}{\ensuremath{\mathfrak R}}

\newcommand{\numBits}{\ensuremath{\mathit{numElems}}}
\newcommand{\setBits}{\ensuremath{\mathit{totalSum}}}
\newcommand{\lastBit}{\ensuremath{\mathit{oldest_\newInputLetter}}}
\newcommand{\outBits}{\ensuremath{\mathit{out}}}
\newcommand{\outBitsVal}{\ensuremath{\parentheses{\nu - \parentheses{(i-o)\mod \nu}}}}

\newcommand{\bs}{{\sc Basic-Summing}}

\newcommand{\add}  [1][] { {\sc Add}$(#1)$}

\newcommand{\window}{n}
\newcommand{\logw}{\log \window}

\newcommand{\blockOffset}{o}
\newcommand{\inputVariable}{x}

\newcommand{\bsrange}{ \ell }

\usepackage{comment}
\usepackage{url}
\usepackage{dsfont}
\usepackage{bm}
\usepackage{amsmath,amssymb}
\usepackage{float}
\usepackage{multirow}
\usepackage{hhline}
\usepackage{calc}
\usepackage{amsthm}

\usepackage{graphicx}
\usepackage{framed}
\usepackage{xspace}
\usepackage{algorithm}
\usepackage{array}

\newtheorem{thm}{Theorem}[section] 

\theoremstyle{plain} 
\newcommand{\thistheoremname}{}
\newtheorem{genericthm}[thm]{\thistheoremname}

\renewcommand{\qed}{\nobreak \ifvmode \relax \else
	\ifdim\lastskip<1.5em \hskip-\lastskip
	\hskip1.5em plus0em minus0.5em \fi \nobreak
	\vrule height0.75em width0.5em depth0.25em\fi}

\def\_{\,\,\,\,\,}

\newcommand{\eps}{\epsilon}

\newcommand{\set}[1]{\left\{#1\right\}}

\usepackage{multicol}
\usepackage{amsthm}
\usepackage{amsmath}

\usepackage{pifont}
\algtext*{EndFunction}
\algtext*{EndIf}

\allowdisplaybreaks
\begin{document}
	
\crefformat{footnote}{#2\footnotemark[#1]#3}

\title{Succinct Approximate Rank Queries}

%
\date{}
\author{\ \\Ran Ben Basat\\
	Department of Computer Science, Technion\\
	\texttt{sran@cs.technion.ac.il}}
\authorrunning{Ran Ben-Basat} 
%
%
%
%
%
%
%
\maketitle
%
\abstract{%
	We consider the problem of summarizing a multi set of elements in $\orange{n}$ under the constraint that no element appears more than $\ell$ times. The goal is then to answer \emph{rank} queries --- given $i\in\orange{n}$, how many elements in the multi set are smaller than $i$? --- with an additive error of at most $\lnrErrSymbol$ and in constant time.
	For this problem, we prove a lower bound of $\lnrLBSymbol \triangleq \floor{\frac{n}{\ceil{\lnrErrSymbol / \ell}}} \log\big({\max\set{\floor{\ell / \lnrErrSymbol},1} + 1}\big)$ 
	bits and provide a \emph{succinct} construction that uses $\lnrLBSymbol\sFactor$ bits.
	Next, we generalize our data structure to support processing of a stream of integers in $\frange{\ell}$, where upon a query for some $i\le n$ we provide a $\lnrErrSymbol$-additive approximation for the sum of the \emph{last} $i$ elements.
	We show that this too can be done using $\lnrLBSymbol\sFactor$ bits and in constant time. 
	This yields the first sub linear space algorithm that computes approximate sliding window sums in $O(1)$ time, where the window size is given at the query time; additionally, it requires only $\sFactor$ more space than is needed for a fixed window size.
}
\section{Introduction}
\label{sec:intro}

\subsection{Background}
Static dictionaries are data structures that encode a set $S\subseteq\orange{n}$ and efficiently answer \emph{membership queries} of the form ``is $i\in S$?'' (for some $i\in\orange{n}$).
This problem was extensively studied and memory efficient data structures that allow $O(1)$ time queries for it were suggested for several different models~\cite{brodnik1999membership,fredman1984storing,pagh2001low}. 

An extension of the dictionary problem is the \emph{Rank} query, which given an identifier $i\le n$ returns the number of elements in $S$ that are smaller than or equal to $i$. For this problem as well, multiple papers proposed space efficient solutions with a constant query time~\cite{Jacobson,Raman2007,Raman99staticdictionaries}.
The inverse problem, called \emph{Select query}, asks for the ID of the $i^{th}$ smallest element in $S$ and was also shown to have space efficient data structures that support constant time queries~\cite{clark1998compact, Raman2007}.

A seemingly different research area is the design of streaming algorithms.
For many domains, such as networking, economics and databases, the ability to process large data streams is vital.
As data varies over time, recent data is often considered more relevant; this motivated the study of \emph{sliding window} algorithms, in which only the last $n$ elements are of interest.
The sliding window model was studied for many problems such as summing~\cite{SWATPAPER,DatarGIM02,GibbonsT02}; counting the number of distinct elements~\cite{basat2017efficient,Fusy2007}; finding frequent elements~\cite{WCSS,HungLT10}; answering set membership queries~\cite{SWAMP,SlidingBloomInfocom,SlidingBloomFilter}; and other problems~\cite{BabcockDMO03,Graphs,TopKSlidingWindow,TopKPairsSlidingWindows}.
All these works share a common goal -- they significantly reduce the memory consumption; in return, they settle for approximate, rather than exact, solutions.
Given sufficient space, we can solve such problems exactly simply by adding the newly arriving item into our summary and deleting the element that has left the window.
However, in many applications the window size is too large and the memory requirement becomes a major bottleneck. In this paper, we show how rank queries can be used for streaming.

Even if modern RAM memories seem to be enough for storing large element sequences, there are many advantages in minimizing the memory requirements.
Routers, for example, often rely on the scarce SRAM which allows access at the speed in which they are required to route packets. If the measurement algorithms are not compact enough to fit into the small SRAM, they must access the slower DRAM that does not allow real time queries. This can be a significant limitation for applications that require timely insights about the traffic, such as load balancing or denial of service attack identification. Similarly, when implementing in software we can gain speed if we fit our algorithm into the CPU cache and reduce DRAM access. Smaller data structures might even fit in a single cache line and can be pinned there to maximize the measurement performance.

The works mentioned above significantly reduce the space requirements compared to storing the entire window in memory.
However, these algorithms assume that the window size is known in advance and their data structures only allow queries about the predetermined window size to be answered efficiently.
While we can maintain a different sketch for every window size that is of interest, this may be prohibitively expensive in terms of both memory and update time.
Further, the goal of these algorithms is to enable memory feasible solutions to what would otherwise require storing the exact window in memory; 
thus, duplicating the data structures for multiple window sizes undermines the purpose for which they were created. 


\subsection{Our Contributions}
Our first contribution is the extension of exact succinct rankers to multi sets in which every element can appear at most $\ell$ times. Previous works have considered multi sets under cardinality constraint for all elements combined. Here we address the natural case where every element may appear at most $\ell$ times, but no cardinality constraint (smaller than $n\cdot\ell$) is known for the multi set. Our approach requires $\sFactor n\logp{\ell+1}$ bits and allows $O(1)$ time rank queries.

Our next contribution are novel \emph{approximate} set and multi set representations that allow computing rank queries with an additive error of $\lnrErrSymbol$, while using less space than required for storing the multi set itself. 
For this problem, we prove a $\lnrLBSymbol = \floor{\frac{n}{\ceil{\lnrErrSymbol / \ell}}} \log\big({\max\set{\floor{\ell / \lnrErrSymbol},1} + 1}\big)$ 
bits lower bound and a propose a succinct data structure that uses $\lnrLBSymbol\sFactor$ bits. 
To the best of our knowledge, this is the first algorithm that provides approximate rank queries in $O(1)$ time using less memory than the set / multi set encoding requires.

Next, we extend the notion of approximate rankers to streams and propose algorithms that process a stream of integers in $\frange{\ell}$ and answer sliding window sum queries in $O(1)$ time. Unlike previous works~\cite{basat2017efficient,SWATPAPER,DatarGIM02}, we get the window size \emph{at query time}. That is, our algorithm can compute the sum of \emph{any} window size while previous works assume that the size is fixed. Interestingly, our construction is succinct even when compared with the lower bound derived in~\cite{SWATPAPER} for fixed size windows. Thus, with a $\sFactor$ space overhead we allow the algorithm to support all window sizes. 
This is a major improvement over the naive approach of maintaining a separate algorithm instance for every window size that is of interest, in both space and~time~complexity.

We note that our approach also allows approximating the sum of \emph{historical intervals} that can be used for drill-down queries. For example, assume that we are monitoring a 100Gbps link on a backbone router such that at each second we get the utilized bandwidth (i.e., we can set $\ell=100\cdot2^{30}$ bits). Now, assume that we identify a distributed denial of service attack and want to study the link utilization pattern before and during the attack. Our algorithm allows us to estimate the bandwidth between any time interval $t_1-t_2$ (for $t_2\le t_1\le n$) simply by subtracting the estimate for the sum of the last $t_2$ seconds from the estimate of the last $t_1$ seconds' sum. 

\section{Related Work}
\subsection{Dictionaries}
Consider a set $S\subseteq\orange{n}$. A \emph{dictionary} is a data structure that supports \emph{membership queries} of the form ``Is $x$ in $S$?''. Several hashing-based works proposed methods for efficiently encoding $S$ while supporting constant time membership queries~\cite{brodnik1999membership,fredman1984storing,pagh2001low,tarjan1979storing}.
Dictionaries were then naturally extended to the \emph{Indexable Dictionary} problem that also supports the operations:
\begin{enumerate}
	\item \textbf{Rank$(i)$}: given $i\in\orange{n}$, return $|\set{y\in S:y\le i}|$.
	\item \textbf{Select$(i)$}: given $i\in\orange{|S|}$, return the $i^{th}$ smallest element in $S$.
\end{enumerate}
The problem of storing sets (and multisets, with the appropriate generalizations of the Rank and Select procedures) drew lots of attention from the research community~\cite{Jacobson,Raman2007,Raman99staticdictionaries}.
Of special interest to us is the work of Jacobson~\cite{Jacobson} that allows constant time rank queries using $n+o(n)$ memory. Jacobson's idea was to look at the characteristic vector of the set, i.e., a $\set{0,1}^n$ bits vector whose $i^{th}$ entry is set if $i\in S$. Thus, the Rank query reduces to counting the number of set bits that precede some index $i$ given at query time. To achieve this, Jacobson breaks the vector into $(\log n)^2$ sized \emph{chunks}. At the end of each chuck, Jacobson keeps the number of set bits that precede it. Since there are $n/(\log n)^2$ such chucks, and each is encoded using $\log n$ bits, this requires $n/\log n = o(n)$ bits. Next, Jacobson focuses on each specific chunk and divides it into a sequence of $(1/2\cdot\log n)$-sized \emph{sub-chunks}. At the end of each sub-chunk, Jacobson stores its number of preceding set bits \emph{within the current chunk} using $O(\log\log n)$ bits. Once again, the number of sub-chunks is $n/(1/2\cdot\log n)$ so the total memory required is $O(n \frac{\log\log n}{\log n})=o(n)$ bits. Finally, Jacobson counts the number of set bits within each sub-chunk using a lookup table. In the table, the keys are all binary vectors of size at most $1/2\cdot\log n$ and the values are the number of set bits; thus, the table's overall memory consumption is $O(\sqrt n \log n \log\log n) = o(n)$.

In this paper, we present a succinct structure for rank queries of multi sets in which each element appears at most $\ell$ times. This is different than the multi set representations of~\cite{pagh2001low,Raman2007} that considered cardinality constraint for the entire multi set, but without any further restriction on the number of appearances of a single item. 
We also provide an encoding that supports additive approximations of rank queries in less memory than required for encoding the multi set. 

\subsection{Algorithms that Sum over Sliding Windows}
Approximating the sum of the last $n$ elements over an integer stream, known as \bs{}, was first introduced by Datar et al.~\cite{DatarGIM02}. 
They assumed that each element is in $\frange{\bsrange}$ and proposed a $(1+\eps)$ multiplicative approximation algorithm.
Their data structure, named {Exponential Histogram} $(\mathit{EH})$,  is based on keeping timestamps of element sequences called buckets such that the last $n$ elements fit into $O(\oneOverE\logp{\bsrange\cdot n})$ buckets. Each bucket requires $O(\log n)$ bits to store the timestamp in addition to $O(\log{\logp{\bsrange\cdot n}})$ bits to store the bucket size.
Overall, the number of bits required by their algorithm is $O\parentheses{\oneOverE\parentheses{\log^2 n +\log\bsrange\cdotpa{\logw+\log\log\bsrange}}}$ and it operates in amortized time $\Op{\frac{\log\bsrange}{\logw}}$ or $O(\log (\bsrange\cdot n))$ worst case.
The EH approach was then extended in~\cite{BabcockDMO03} for other statistics over sliding windows, such as median and variance.
In~\cite{GibbonsT02}, Gibbons and Tirthapura presented a $(1+\eps)$ multiplicative algorithm that operates in constant worst case time while using similar space for $\bsrange=\window^{O(1)}$.
In~\cite{SWATPAPER}, we studied the potential memory savings one can get by replacing the $(1+\eps)$ multiplicative guarantee with a $\lnrErrSymbol$ additive approximation. We showed that $\Thetap{\frac{\bsrange\cdot n}{\lnrErrSymbol}+\logw}$ bits are required and sufficient.

In a sliding window, the last $n$ elements get similar weight while older items do not affect the sum.
Cohen and Strauss~\cite{CohenS06} considered more general aging models where older data has lower weight, but the rate in which the weight decreases may be different than that of sliding~windows.

Recently, we studied~\cite{basat2017efficient} the affect that allowing an error in the \emph{window size} has on the required memory of approximate summing algorithms. Specifically, we showed that if upon a query the algorithm is required to return a tuple $\langle w, \widehat{S_w}\rangle$ such that $w\in\{\window,\window+1,\ldots,\window(1+\tau)\}$ and $|\widehat {S_w} - S_w| < \lnrErrSymbol$ then $\Thetap{\oneOverT\logp{\frac{\tau\cdot\bsrange\cdot n}{\lnrErrSymbol}}+\logw}$ bits are needed.

All of the algorithms above assume that the window size is fixed. 
Here, we propose solutions that are \emph{succinct}, even when compared to a lower bound derived here for static data, or to the bound for a fixed size window as in~\cite{SWATPAPER}. 

It is worth mentioning that these data structures \emph{do} allow computing the sum of a window whose size is given at the query time. Alas, the query time will be slower as they do not keep aggregates that allow quick computation. Specifically, we can compute a $(1+\epsilon)$ multiplicative approximation using a slightly extended version of EH~\cite{DGIM02} in $O(\log\oneOverE+\log\log n)$ time by a binary search for the block with the right timestamp. We can also use the data structure of~\cite{SWATPAPER} for an additive approximation of $\lnrErrSymbol$  in $O\parentheses{\min\set{\frac{\ell\cdot\window}{\lnrErrSymbol},n}}$ time, and utilize~\cite{basat2017efficient}'s structure for a $(\tau,\lnrErrSymbol)$-approximation in time $O(\oneOverT)$. 
In this paper we offer solutions that operate in $O(1)$~time.

\section{Preliminaries}

We say that an algorithm is \emph{succinct} if it uses $\mathcal B(1+o(1))$ bits, where $\mathcal B$ is the information-theoretic lower bound for the problem it solves. Throughout the paper, we assume the standard word RAM model with a word size of $\Thetap{\log n + \log \ell}$. For simplicity of presentation, we also assume that ${n / (\log n)^2}$ and $\sqrt{\log n}$ are integers.

\begin{definition}[Approximation]
	Given a value $V$ and a constant $\epsilon>0$, we say that $\widehat{V}$ is an \emph{$\epsilon$-additive approximation} of $V$ if $V - \eps < \widehat{V} \le  V$.\footnote{We use one-sided error, and strict inequality as this simplifies our computations.}
\end{definition}
%
%
%
Next, we define the notion of an \lnr{} -- a structure that can answer approximate rank queries in a memory efficient manner.
Specifically, \lnr[\minDelta] is a succinct encoding of a multi-set over $\frange{n},$ such that no element appears more than $\ell$~times, that supports $O(1)$ time rank queries. 

\begin{definition}[Static Ranker]
An \lnr, for some $\ell,n,\lnrErrSymbol\in\mathbb N^+$, is an algorithm that preprocesses a sequence in $\set{0,1,\ldots,\ell}^n$ and when queried with some $i\le n$ returns a $\lnrErrSymbol$-additive approximation $\widehat{S_i}$ of the sum of the \emph{first} $i$ elements, ${S_i}$, in $O(1)$ time. 
\end{definition}
We proceed with the definition of a Sliding Ranker, extending \lnr{}s to streams, while focusing on the \emph{last} elements in the stream for supporting sliding window queries.
\begin{definition}[Sliding Ranker]
An \qsr[\lnrErrSymbol], for some $\ell, n,\lnrErrSymbol\in\mathbb N^+$, is an algorithm that processes a stream of integers in $\frange{\ell}$ and when queried for some $i\le n$ returns
a $\lnrErrSymbol$-additive approximation $\widehat{S_i}$ of the \emph{last} $i$ elements sum, ${S_i}$, in $O(1)$ time. 
\end{definition}
%



\section{\lnr{}s}
In order to construct an \lnr{}, we first discuss the special case of zero-error ($\lnrErrSymbol=\minDelta$).
\subsection{An \lnr[\minDelta]}
Here, we provide a succinct construction of \lnr[\minDelta], for any $\ell,n$.
Intuitively, this generalizes Jacobson's ranker~\cite{Jacobson} that addresses binary sequences ($\ell=1$). As we show, his lookup table approach works for ``small'' values of $\ell$. 
In other cases, such as $\ell \ge n$, we can split the vector into smaller and smaller intervals (i.e., sub-sub-chunk, etc.), but if the number of levels is constant, storing a lookup table for the smallest level is infeasible in $o(\mathcal B)$ space.
Thus, we use a different trick for large $\ell$ values for computing within-sub-chunk sums in $O(1)$. 
We avoid keeping the \emph{and} the characteristic vector; instead, we keep a $n$-sized array in which each entry contains the sum of the sub-chunk up to that point. For example, if the sub-chunk was $\langle 1,0,1\rangle$, we store $\langle 1,1,2\rangle$ regardless of vector entries outside this sub-chunk. Since $\ell$ is ``large'', this takes $o(\mathcal B)$ space.

We start by noting that since the number of sequences in $\frange{\ell}^n$ is $(\ell+1)^n$, any algorithm that computes such rank queries (exactly) requires $\mathcal B \triangleq{n\log (\ell+1)}$ bits. We also note that without a query time constraint this is achievable, as we can simply store the entire data and when queried sum the required interval in $O(n)$ time. 
Thus, if $n=O(1)$ (i.e., we have a small array of potentially large numbers), then the same idea works and we therefore require $\mathcal B$ bits; hence, we hereafter assume that 
\ensuremath{n = \omega(1)}.
Next, we will prove the following:
\begin{theorem}\label{thm:exactRanker}
For any $\ell,n\in\mathbb N^+$, there exists an \lnr[\minDelta] that uses $\mathcal B\smallMultError$ bits.
\end{theorem}	
We start by breaking the sequence into chunks of size $\log^2 n$, keeping the cumulative sums at the end of each chunk. The required number of bits for these sums is at most 
\begin{align*}
\frac{n}{\log^2 n} \logp{\ell\cdot n + 1} \le \frac{n}{\log^2 n} \logp{(\ell+ 1)\cdot n} = n\logp{ \ell + 1} \cdotpa{\frac{1}{\log^2 n} +\frac{1}{\log \ell\log n} } = o(\mathcal B),
\end{align*}
where the last equation follows from \ensuremath{n = \omega(1)}.

Next, we break the chunks into sub-chunks of size $\sqrt{\log n}$ and keep the cumulative sum \emph{from the beginning of the most recent chunk} at the each sub-chunk's end. The memory consumption of these sub-chunk aggregates is then no more than
\begin{align*}
\frac{n}{\sqrt{\log n}} \logp{\ell\cdot \logp{\log^2 n} + 1} \le n\logp{ \ell + 1} \cdotpa{\frac{1}{\sqrt{\log n}} +\frac{2\log \log n}{\log \ell\sqrt{\log n}}} = o(\mathcal B).
\end{align*}

We are left with the task of efficiently computing the sub-chunk sums. Here, we split our construction depending on the relation between $\ell$ and $n$.
\begin{itemize}
\item {\bm{$\ell+1 \le 2^{\sqrt[3]{\log n}}$}.}\\
In this case, we adopt Jacobson's lookup table approach. Specifically, we create a lookup table $T:\frange{\ell}^{\sqrt{\log n}}\times \frange{\sqrt{\log n}-1}\to \frange{\ell\cdot\sqrt{\log n}}$; the key of each table entry is a $\sqrt{\log n}$-sized sequence of elements in $\frange{\ell}$ and an index $k\in\frange{\sqrt{\log n}-1}$. Its value is the sum of the first $k$ sequence entries.
In order to use the table, we also store the characteristic vector itself using $n\logp{\ell+1}$ bits.
The size of the table is then
\begin{align*}
(\ell+1)^{\sqrt{\log n}} \sqrt{\log n} \logp{\ell\sqrt{\log n}+1} \le 2^{\log^{5/6} n +\log\log n}  \logp{(\ell+1)\sqrt{\log n}} = o(\mathcal B).
\end{align*}
Thus our overall memory consumption is $n\logp{\ell+1}\cdot\smallMultError$.
Unfortunately, while we can consider smaller and smaller sequence aggregates, constructing such a lookup table will prevent the algorithm from being succinct when $\ell$ is large (e.g., for $\ell\ge n$). 
\item {\bm{$\ell+1 > 2^{\sqrt[3]{\log n}}$}.}\\
In this case, we return to the cumulative approach. \emph{Instead} of storing the characteristic vector (and without a lookup table) we store {for each element} the cumulative sum \emph{from the beginning of its sub-chunk}. Since the sub-chunks are of size $\sqrt{\log n}$, the number of bits this takes is
\begin{align*}
\hspace{-1cm}n\logp{\ell\sqrt{\log n}+1} \le n\logp{ \ell + 1}\cdotpa{1+\frac{\log\log n}{\logp{\ell+1}}} \le n\logp{ \ell + 1}\cdotpa{1+\frac{\log\log n}{\sqrt[3]{\log n}}} = \mathcal B+o(\mathcal B).
\end{align*}
\end{itemize}
We conclude that in all cases our construction requires $\mathcal B\smallMultError$ bits and is thus succinct.
\subsection{An \lnr{} for $\lnrErrSymbol>\minDelta$}\label{sec:lnr}
We start by proving a lower bound on the memory required by any \lnr{}. For convenience, we denote $\mu\triangleq \lnrErrSymbol / \ell$. We only consider $\lnrErrSymbol\in\set{2,\ldots,\ell\cdot n}$, as $\lnrErrSymbol=1$ means zero-error and $\lnrErrSymbol> n\cdot\ell$ allows the algorithm to always return $0$, regardless of the input.
\begin{theorem}
Let $\ell,n,\lnrErrSymbol\in\mathbb N^+$, then the number of bits required by any deterministic \lnr{} is at least
$$\lnrLBSymbol\triangleq \lnrLB.$$
\end{theorem}
\begin{proof}
We denote 
$I\triangleq\set{\min\set{\lnrErrSymbol\cdot k,\ell} \mid k\in \frange{\bContentSize}}\subseteq\frange{\ell}$ and $\bar{I}\triangleq\set{\sigma^{\ceil{\mu}}\mid\sigma\in I}$. Next, consider all inputs that contain a sequence of $\nBlocks$ \emph{blocks} padded by zeros, such that each block is a member of $\bar I$; that is, consider 
$\mathcal I \triangleq \bar I^{\nBlocks}\cdot 0^{n\mod \ceil{\mu}}$. Notice that each literal is in the range $\frange{\ell}$ and that each input is of size $n$ as required. We show that every two inputs in $\mathcal I$ must lead to distinct configurations in the \lnr{}, thereby implying a $\ceil{\log |\mathcal I|}$ bits lower bound as required. Let 
$x_1=x_{1,1}x_{1,2}\cdots x_{1,\nBlocks}0^{n-(n\mod \ceil{\mu})},\ 
x_2=x_{2,1}x_{2,2}\cdots x_{1,\nBlocks}0^{n-(n\mod \ceil{\mu})}$ be two distinct inputs in $\mathcal I$ such that $x_{\alpha,\beta}\in \bar I$ for any $\alpha\in\set{1,2},\beta\in\set{1,\ldots,\nBlocks}$. Denote by $t\triangleq\min\set{\gamma\in\set{1,\ldots,\nBlocks}\mid x_{1,\gamma}\neq x_{2,\gamma}}$ the first block's index in which $x_1$ differs from $x_2$. 
Now consider a query for $i\triangleq \bSize\cdot t$. If $\mu \le 1$, then $\nBlocks=n$ and (due to the definition of $I$) $|x_{1,t}-x_{2,t}|\ge \lnrErrSymbol$, which implies an error of at least $\lnrErrSymbol$ for at least one of the inputs. On the other hand, $\mu > 1$ means that $I = \set{0,\ell}$ and thus either $x_{1,t}=0^{\bSize}, x_{2,t}=\ell^{\bSize}$ or $x_{1,t}=\ell^{\bSize}, x_{2,t}=0^{\bSize}$. In either case, the difference in sums is at least ${\bSize}\cdot\ell \ge \lnrErrSymbol$. We established that if two inputs in $\mathcal I$ lead to the same configuration, the error for one of them would be at least $\lnrErrSymbol$ while we assumed it is strictly lower.
\end{proof}
We now present a succinct construction of an \lnr{}.
Denote $\nu \triangleq \UBnItemsPerBlock, s\triangleq \UBnBlocks$ and $z\triangleq\floor{\mu^{-1}\nu}$.
For creating an \lnr{}, we first show how to ``compress'' the input into a smaller problem that we solve exactly. 
Intuitively, we create a new $s$-long input \newInput, such that each of its elements is bounded by $z$, and then employ a \ensuremath{(z,s,\minDelta)}-\emph{Ranker}, $\mathcal R$. Alas, if $\mu=\omega(1)$, this is not enough to allow succinct encoding; for this, we also compute the fraction of the input's sum that is not accounted for in \newInput{} and use it for answering queries.
Given an input $\bar x\in\frange{\ell}^n$, we create \newInput{} iteratively as follows\footnote{\label{note1}If $(n\mod\nu) \neq 0$, we implicitly define $\newInputLetter_{\ceil{\frac{n}{\nu}}}\triangleq0$ and $\mathcal R.Query\parentheses{\ceil{\frac{n}{\nu}}} \triangleq \mathcal R.Query\parentheses{\floor{\frac{n}{\nu}}}.$ }:
$$\forall k\in\set{1,2,\ldots,s}:\quad \newInputLetter_k \triangleq\floor{\lnrErrSymbol^{-1}\cdot{\sum_{d=n-\nu\cdot k+1}^{n}x_d}}-\sum_{\jmath=1}^{k-1} \newInputLetter_\jmath
.$$
Then, we compute the remainder:
\begin{equation}
\mathfrak r \triangleq{{\sum_{d=1}^{n}x_d}}-\lnrErrSymbol\cdot\sum_{\jmath=1}^{s} \newInputLetter_\jmath
. \label{eq:rem}
\end{equation}
After computing $\newInput\in\frange{z}^s$, we feed it into a \ensuremath{(z,s,\minDelta)}-\emph{Ranker} denoted $\mathcal R$. Given a query for some $i\le n$, we return\cref{note1} 
$${\sc Query}(i)\triangleq  \mathfrak r - \parentheses{\lnrErrSymbol - 1/2}+\lnrErrSymbol\cdotpa {\sum_{\jmath=\floor{\frac{n-i}{\nu}}+1}^s \newInputLetter_\jmath} -\ell\cdot\newInputLetter_{\ceil{\frac{n-i}{\nu}}}\cdot{(n-i)\mod \nu}\quad,$$ 
which we can compute in $O(1)$ as follows:
\footnote{We note that if our ranker $\mathcal R$ was originally constructed to compute the sum of the \emph{last} $i$ elements rather than the first, only two queries were needed.}
\begin{multline}
{\sc Query}(i)\triangleq  \mathfrak r - \parentheses{\lnrErrSymbol - 1/2} + \lnrErrSymbol\cdotpa {\mathcal R.Query\parentheses{s} - \mathcal R.Query\parentheses{\floor{\frac{n-i}{\nu}}}} \\
- \ell\cdotpa{(n-i)\mod \nu}\cdotpa {{\mathcal R.Query\parentheses{\ceil{\frac{n-i}{\nu}}} - \mathcal R.Query\parentheses{\floor{\frac{n-i}{\nu}}}}}.
\end{multline}
\begin{lemma}\label{lem:lnrCorrectness}
$$\parentheses{\sum_{d=1}^i x_d} - \lnrErrSymbol < {\sc Query}(i) < \sum_{d=1}^i x_d.
$$
 \end{lemma}
\begin{proof}
We denote  the error in the representation of the \emph{last} $\nu\floor{\frac{n-i}{\nu}}$ items by
$$
\xi \triangleq \sum_{d=n-\nu\floor{\frac{n-i}{\nu}}+1}^{n} x_d - \lnrErrSymbol\cdot\sum_{\jmath={1}}^{\floor{\frac{n-i}{\nu}}} \newInputLetter_\jmath 
= \sum_{d=n-\nu\floor{\frac{n-i}{\nu}}+1}^{n} x_d - \lnrErrSymbol\floor{\lnrErrSymbol^{-1}\cdot{\sum_{d=n-\nu\floor{\frac{n-i}{\nu}}+1}^{n}x_d}}.
$$ 
Observe that $\xi = \parentheses{\sum_{d=n-\nu\floor{\frac{n-i}{\nu}}+1}^{n} x_d \mod \lnrErrSymbol}$ and~hence
\begin{align}
0\le\xi\le\lnrErrSymbol-1 \label{eq:errBounds}.
\end{align} Next, we use~\eqref{eq:rem} to obtain
\begin{align}
\sum_{d=1}^{i} x_d &= \sum_{d=1}^{n} x_d - \sum_{d=i+1}^{n} x_d = \mathfrak r  +\notag \lnrErrSymbol\cdot\sum_{\jmath=1}^{s} \newInputLetter_\jmath - \sum_{d=i+1}^{n} x_d \\
&= \mathfrak r  + \lnrErrSymbol\cdotpa{\sum_{\jmath=1}^{\floor{\frac{n-i}{\nu}}} \newInputLetter_\jmath + \sum_{\jmath=\floor{\frac{n-i}{\nu}}+1}^{s} \newInputLetter_\jmath} - \sum_{d=i+1}^{n-\nu\floor{\frac{n-i}{\nu}}} x_d - \sum_{d=n-\nu\floor{\frac{n-i}{\nu}}+1}^{n} x_d \notag\\
&= \mathfrak r - \xi + \lnrErrSymbol\cdot{\sum_{\jmath=\floor{\frac{n-i}{\nu}}+1}^{s} \newInputLetter_\jmath} - \sum_{d=i+1}^{n-\nu\floor{\frac{n-i}{\nu}}} x_d \notag\\
&= {\sc Query}(i) - \xi - \sum_{d=i+1}^{n-\nu\floor{\frac{n-i}{\nu}}} x_d + \lnrErrSymbol - 1/2 + \ell\cdot\newInputLetter_{\ceil{\frac{n-i}{\nu}}}\cdot{(n-i)\mod \nu}.\label{eq:complex}
\end{align}
We now perform a case analysis, based on the value of $\mu$ and start with the simpler case where $\mu < 2$. In this case, we have $\nu=1$ and thus we can rearrange \eqref{eq:complex} as:
\begin{align*}
{\sc Query}(i) - \sum_{d=1}^{i} x_d = \xi - \lnrErrSymbol + 1/2,
\end{align*}
and using \eqref{eq:errBounds} we immediately get $\parentheses{\sum_{d=1}^i x_d} - \lnrErrSymbol < {\sc Query}(i) < \sum_{d=1}^i x_d$.

Next, we focus on the case of $\mu \ge 2$. Thus, we hereafter have $\nu=\floor{\mu}$ and $\forall \jmath\in\orange{s}: \newInput_\jmath\in\set{0,1}$. We now consider if and when both $\newInputLetter_{\ceil{\frac{n-i}{\nu}}} = 1$ \emph{and} $(n-i)\mod \nu \neq 0$ (which implies $\ceil{\frac{n-i}{\nu}}=\floor{\frac{n-i}{\nu}}+1$); observe that 
{	\footnotesize
\begin{align}\hspace{-0.7cm}
\newInputLetter_{\ceil{\frac{n-i}{\nu}}} &=\floor{\lnrErrSymbol^{-1}\cdot{\sum_{d=n-\nu\cdot \ceil{\frac{n-i}{\nu}}+1}^{n}x_d}}-\sum_{\jmath=1}^{\ceil{\frac{n-i}{\nu}}-1} \newInputLetter_\jmath
=\floor{\lnrErrSymbol^{-1}\cdot{\sum_{d=n-\nu\cdot \ceil{\frac{n-i}{\nu}}+1}^{n}x_d}-\sum_{\jmath=1}^{\floor{\frac{n-i}{\nu}}} \newInputLetter_\jmath}\notag\\
&=\floor{\lnrErrSymbol^{-1}\cdot{\sum_{d=n-\nu\cdot \ceil{\frac{n-i}{\nu}}+1}^{n}x_d}+\lnrErrSymbol^{-1}\cdotpa{\xi - \sum_{d=n-\nu\floor{\frac{n-i}{\nu}}+1}^{n}} x_d}
=\floor{\lnrErrSymbol^{-1}\cdotpa{\xi + \sum_{d=n-\nu\floor{\frac{n-i}{\nu}}+1}^{n-\nu\cdot \ceil{\frac{n-i}{\nu}}}x_d}}.\notag
\end{align}
}\normalfont	
Thus, if $(n-i)\mod \nu\neq 0$, then 
\begin{align}
\newInputLetter_{\ceil{\frac{n-i}{\nu}}} = 1 \iff \xi + \sum_{d=n-\nu\floor{\frac{n-i}{\nu}}+1}^{n-\nu\cdot \ceil{\frac{n-i}{\nu}}}x_d \ge \lnrErrSymbol \iff 
\xi + \sum_{d=i+1}^{n-\nu\cdot \ceil{\frac{n-i}{\nu}}}x_d \ge \lnrErrSymbol -\sum_{d= n-\nu\floor{\frac{n-i}{\nu}}+1}^{i}x_d.\label{eq:bitVal}
\end{align}
Next, we split to cases based on the value of $\newInputLetter_{\ceil{\frac{n-i}{\nu}}}$:
\begin{itemize}
	\item $\newInputLetter_{\ceil{\frac{n-i}{\nu}}} = 1$.\quad
	In this case, according to~\eqref{eq:complex} and~\eqref{eq:bitVal} we have:
{	\footnotesize	
	\begin{align*}
		{\sc Query}(i) - \sum_{d=1}^{i} x_d &= \xi + \sum_{d=i+1}^{n-\nu\floor{\frac{n-i}{\nu}}} x_d -\parentheses{\lnrErrSymbol - 1/2 + \ell\cdot{(n-i)\mod \nu}} \\
		&\ge \lnrErrSymbol -\sum_{d= n-\nu\floor{\frac{n-i}{\nu}}+1}^{i}x_d -\parentheses{\lnrErrSymbol - 1/2 + \ell\cdot{(n-i)\mod \nu}}
		\\
		&= -\sum_{d= n-\nu\floor{\frac{n-i}{\nu}}+1}^{i}x_d + 1/2 - \ell\cdot{(n-i)\mod \nu}\\
		&\ge -\sum_{d= n-\nu\floor{\frac{n-i}{\nu}}+1}^{i}\ell + 1/2 - \ell\cdot{(n-i)\mod \nu}\\
		&\ge -(\ell\nu-1) + 1/2 = -(\ell\floor{\Delta/\ell}-1) + 1/2 \ge -\Delta + 1/2.
	\end{align*}
}\normalfont	
	On the other hand, we bound the error from above as follows:
	\begin{align*}
	{\sc Query}(i) - \sum_{d=1}^{i} x_d &= \xi + \sum_{d=i+1}^{n-\nu\floor{\frac{n-i}{\nu}}} x_d -\parentheses{\lnrErrSymbol - 1/2 + \ell\cdot{(n-i)\mod \nu}} \\
	&\le \lnrErrSymbol-1 + \sum_{d=i+1}^{n-\nu\floor{\frac{n-i}{\nu}}} \ell - \lnrErrSymbol - \ell\cdot{(n-i)\mod \nu} + 1/2\le -1/2.
	\end{align*}	
		\item $\newInputLetter_{\ceil{\frac{n-i}{\nu}}} = 0$.\quad
		Similarly to before, using~\eqref{eq:complex} and~\eqref{eq:bitVal} we get:
		\begin{align*}
		{\sc Query}(i) - \sum_{d=1}^{i} x_d &= \xi + \sum_{d=i+1}^{n-\nu\floor{\frac{n-i}{\nu}}} x_d -\parentheses{\lnrErrSymbol - 1/2} < \lnrErrSymbol -\sum_{d= n-\nu\floor{\frac{n-i}{\nu}}+1}^{i}x_d -\parentheses{\lnrErrSymbol - 1/2}\\
		&= -\sum_{d= n-\nu\floor{\frac{n-i}{\nu}}+1}^{i}x_d + 1/2 \le 1/2.
		\end{align*}		
		Now, we use the fact that both ${\sc Query}(i)$ and $\sum_{d=1}^{i} x_d$ are integers to deduce that ${\sc Query}(i) - \sum_{d=1}^{i} x_d < 1/2 \implies {\sc Query}(i) - \sum_{d=1}^{i} x_d \le 0$.
		Finally, we bound the error from above:
		\begin{align*}
		{\sc Query}(i) - \sum_{d=1}^{i} x_d &= \xi + \sum_{d=i+1}^{n-\nu\floor{\frac{n-i}{\nu}}} x_d -\parentheses{\lnrErrSymbol - 1/2} \ge -\lnrErrSymbol + 1/2.
		\end{align*}	
\end{itemize}
We conclude that in all cases we have $\parentheses{\sum_{d=1}^i x_d} - \lnrErrSymbol < {\sc Query}(i) < \sum_{d=1}^i x_d$.
\end{proof}
Next, we show that the value of each entry in $\newInput$ is smaller than $z$, as stated.
\begin{lemma}
For any $k\in\set{1,2,\ldots,s}$, $\newInputLetter_k\le \floor{\mu^{-1}\nu}$.
\end{lemma}
\begin{proof}
Notice that $\newInputLetter_1 = \floor{\lnrErrSymbol^{-1}\cdot{\sum_{d=1}^{\nu}x_d}}\le \floor{\lnrErrSymbol^{-1}\cdot\nu\ell}=\floor{\mu^{-1}\nu}$. For other $k$ values, we have
{	\footnotesize
\begin{align*}
\newInputLetter_k &= \floor{\lnrErrSymbol^{-1}\cdot{\sum_{d=1}^{\nu\cdot k}x_d}}-\sum_{\jmath=1}^{k-1} \newInputLetter_\jmath
= \floor{\lnrErrSymbol^{-1}\cdot{\sum_{d=\nu\cdot (k-1)+1}^{\nu\cdot k}x_d} +\parentheses{ \lnrErrSymbol^{-1}\cdot{\sum_{d= 1}^{\nu\cdot (k-1)}x_d}-\sum_{\jmath=1}^{k-2} \newInputLetter_\jmath}-\newInputLetter_{k-1}}\\
&\le \floor{\lnrErrSymbol^{-1}\cdot{\sum_{d=\nu\cdot (k-1)+1}^{\nu\cdot k}x_d} + \parentheses{\floor{\lnrErrSymbol^{-1}\cdot{\sum_{d= 1}^{\nu\cdot (k-1)}x_d}}-\sum_{\jmath=1}^{k-2} \newInputLetter_\jmath}-\newInputLetter_{k-1}}
\\
&=\floor{\lnrErrSymbol^{-1}\cdot{\sum_{d=\nu\cdot (k-1)+1}^{\nu\cdot k}x_d} }
\le \floor{\lnrErrSymbol^{-1}\cdot\nu\ell}=\floor{\mu^{-1}\nu}.\qedhere
\end{align*}
}
\normalfont
\end{proof}
%
We now bound $\mathfrak r$ for analyzing the space of our construction; the proof appears in Appendix~\ref{apx:remVal}.
\begin{lemma}\label{lem:remVal}
For any input $x\in\frange{\ell}^n$, the remainder in~\eqref{eq:rem} satisfies $\mathfrak r < 2\lnrErrSymbol$.
\end{lemma}
Follows is an analysis of our ranker.
\begin{lemma}
Let $\ell,n,\lnrErrSymbol\in\mathbb N^+$ and $\mu\triangleq\lnrErrSymbol/\ell$. The number of bits required by our ranker is $\sFactor\cdot\floor{n/\max\set{\floor{\mu},1}}\cdot\log\big({\ceil{\mu^{-1}} + 1}\big)$.
\end{lemma}
\begin{proof}
Our construction has two components: the exact ranker $\mathcal R$ and the remainder $\mathfrak r$. As $\mathcal R$ is a \ensuremath{(z,s,\minDelta)}-\emph{Ranker}, where $s\triangleq \UBnBlocks$ and $z\triangleq\floor{\mu^{-1}\nu}$, it requires $\sFactor\cdot{s\logp{z+1}}$ bits according to Theorem~\ref{thm:exactRanker}. Recalling that $\nu=\UBnItemsPerBlock$ gives us the desired $\lnrLBSymbol\sFactor$ bound. Finally, Lemma~\ref{lem:remVal} tells us that $\mathfrak r < 2\lnrErrSymbol$ and can therefore be represented using $O(\log \lnrErrSymbol)= o(\lnrLBSymbol)$~bits.\qedhere
\end{proof}
\begin{theorem}\label{thm:lnrSuccinct}
Let $\ell,n,\lnrErrSymbol\in\mathbb N^+$ such that $(\mu=o(1))\vee (\mu=\omega(1)) \vee (\mu\in\mathbb N) \vee (\mu^{-1}\in\mathbb N)$, 
the construction above is an \lnr{} that uses $\lnrLBSymbol\smallMultError$ bits.\footnote{In other cases, our construction uses at most $B(2+o(1))$ bits but might not be succinct.}
\end{theorem}
\begin{proof}
Recall that $\lnrLBSymbol= \lnrLBVal$ while our algorithm uses $\sFactor\cdot\floor{n/\max\set{\floor{\mu},1}}\cdot\log\big({\ceil{\mu^{-1}} + 1}\big)$ bits. If $\mu=o(1)$, we have $\lnrLBSymbol=n \log\big({\floor{\mu^{-1}} + 1}\big)=\sNegFactor n\log\mu^{-1}$ when our structure takes $\sFactor\cdot n\cdot\log\big({\ceil{\mu^{-1}} + 1}\big)=\sFactor n\log\mu^{-1}$ bits.
Similarly, if $\mu=\omega(1)$ then $\lnrLBSymbol=\floor{n/\bSize}=\sNegFactor\cdotpa{n/\mu}$ while we require $\sFactor\cdot\floor{n/\floor{\mu}}=\sFactor\cdotpa{n/\mu}$. The case for $(\mu=\Theta(1))\wedge((\mu\in\mathbb N) \vee (\mu^{-1}\in\mathbb N))$ follows from similar arguments.
\end{proof}

\section{\qsr[\lnrErrSymbol]s}
As in the case of static data rankers, we first consider the exact case where $\lnrErrSymbol=\minDelta$.
\subsection{An \qsr{}}\label{sec:qsr}
In this section, we provide a construction for an \qsr{} that requires $\mathcal B(1+o(1))$ bits, where $\mathcal B\triangleq n\logp{\ell + 1}$ is the information-theoretic lower bound even without considering sliding windows.
Intuitively, we adapt our \lnr[\minDelta] construction to the sliding window setting by incrementally building the chunks and sub-chunks. We start by breaking the stream into $n$-sized \emph{frames}. As in the original construction, we split the frames into $(\log n)^2$ sized \emph{chunks}, where each chunk is further divided into $\sqrt{\log n}$-sized \emph{sub-chunks}. 
In the case where $\ell+1 \le 2^{\sqrt[3]{\log n}}$, we keep a $O(2^{\log^{5/6} n +\log\log n}  \logp{\ell\sqrt{\log n}})$ sized lookup table that maps each sequence in $\frange{\ell}^{\le \sqrt{\log n}}$ to its sum.
If  $\ell+1 > 2^{\sqrt[3]{\log n}}$, we simply track the sums \emph{within a sub-chunk} by keeping the cumulative sum for each item.
We keep the chunk aggregates in a ${n / (\log n)^2}$-sized circular buffer, and the sub-chunk aggregates in a similar structure of size ${n / \sqrt{\log n}}$. 
Finally, we ``reset'' the frame accumulator every $n$ elements, so that each chunk's aggregate is always smaller than $n\cdot \ell$.
Since each of the chunk aggregates requires $O(\log (\ell n))$ bits and each of the sub-chunk aggregates takes $O(\log(\ell\log n))$ bits, our overall space consumption is as required.
Our \qsr{} construction is illustrated in Figure~\ref{fig:slidingRankerConstruction}, while the query procedure is exemplified in Figure~\ref{fig:slidingRankerQuery}.
In Appendix~\ref{apx:largeRslidingRanker} we provide an algorithm for the $\ell+1 > 2^{\sqrt[3]{\log n}}$ case; here, we hereafter assume that $\ell+1 \le 2^{\sqrt[3]{\log n}}$.
\begin{figure}[]
	\centering
	\includegraphics[width=\linewidth, height=3.0cm]{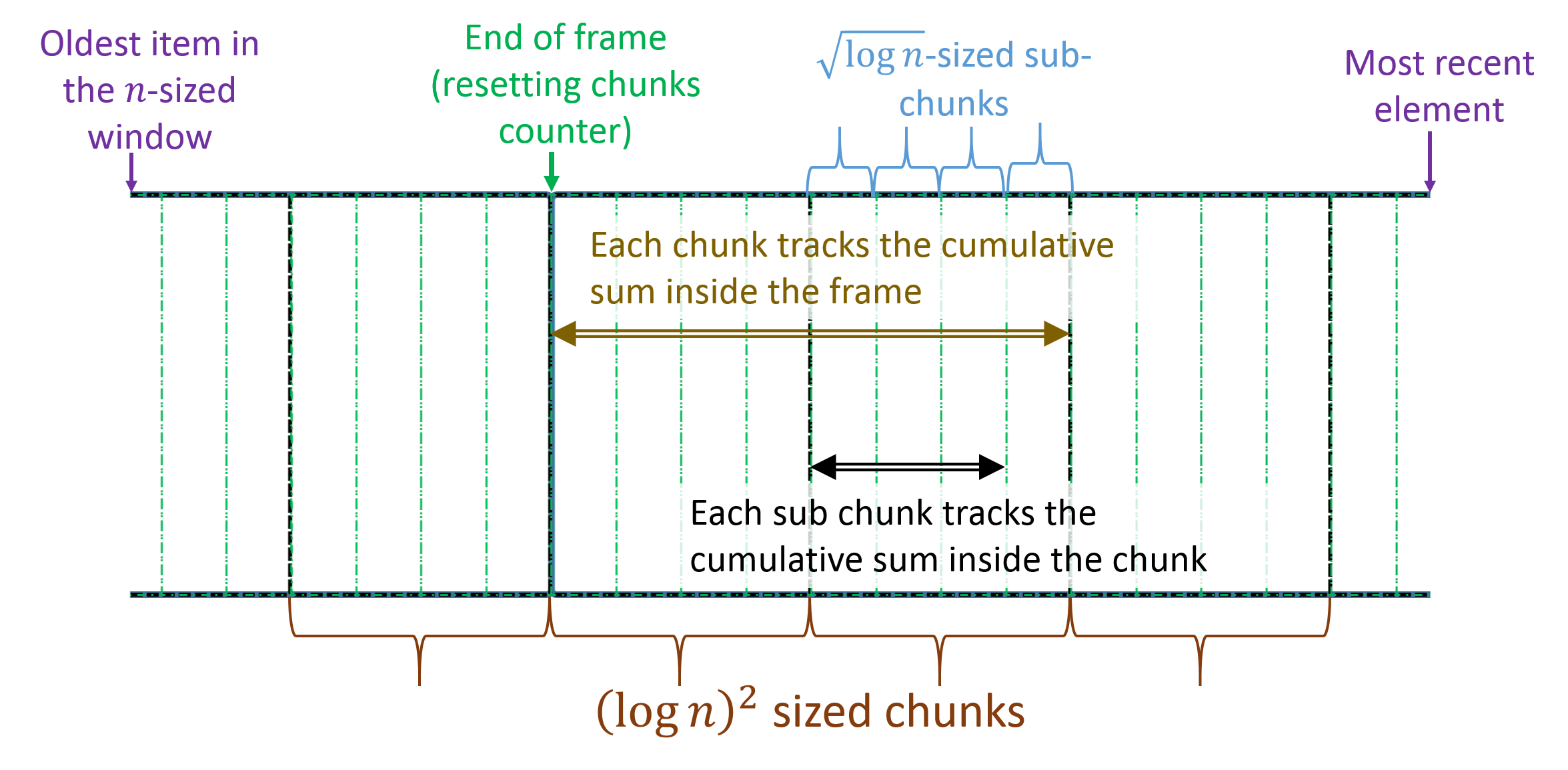}
	\caption{The \qsr{} construction of Algorithm~\ref{alg:qsr}. We split the stream into frames, the frames into chunks, and the chunks into sub-chunks. At the end of each chunk we keep the sum of all elements that preceded it in the frame. Similarly, for each sub-chunk we keep the sum of items from the beginning of its chunk.}
	\label{fig:slidingRankerConstruction}
\end{figure}
\begin{figure}[]
	\centering
	\includegraphics[width=\linewidth, height=2.5cm]{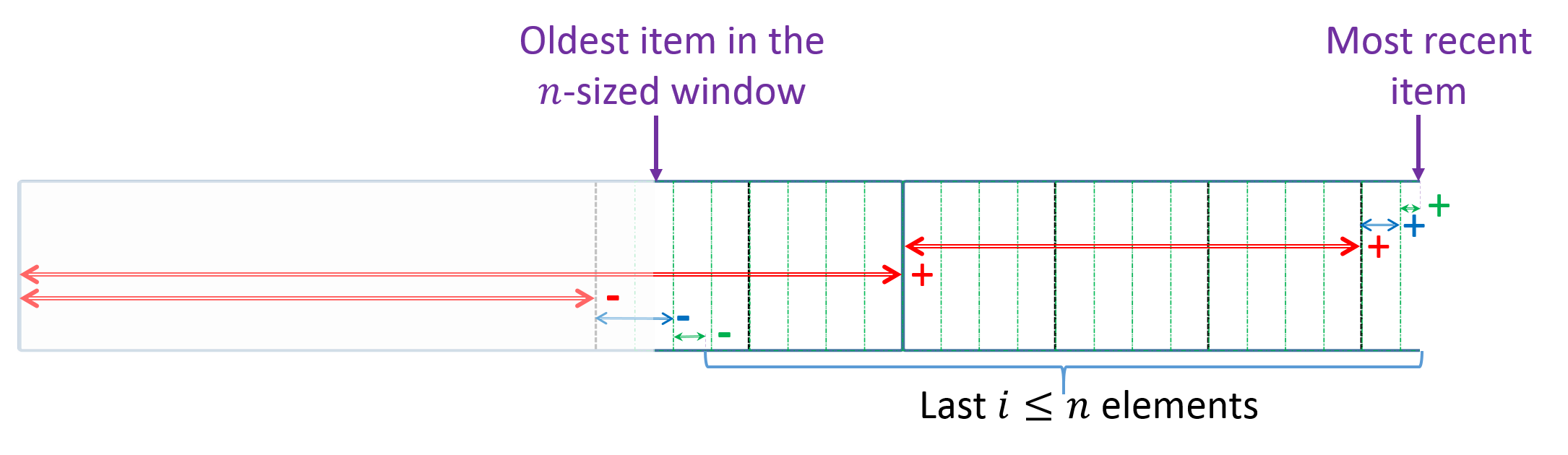}
	\caption{A query example of Algorithm~\ref{alg:qsr}. To compute the sum of the green ranges, we use either a lookup table or the within-sub-chunk aggregates depending on whether $\ell+1 \le 2^{\sqrt[3]{\log n}}$. The sub-chunks aggregate array allows us to retrieve the sum of the blue intervals, and the chunk aggregates contain the sum of the red ranges. We first add values to get the sum of elements from the beginning of the frame in which the $i$-long interval starts. Then, we subtract the sum of items that arrived more than $i$ elements ago.}
	\label{fig:slidingRankerQuery}
\end{figure}
Our algorithm uses the following variables:
\begin{itemize}
\item $C$ - a cyclic buffer of ${n / (\log n)^2}$ integers, each allocated with $\ceil{\log (\ell n + 1)}$ bits.
\item $SC$ - a cyclic buffer of $n/\sqrt{\log n}$ integers, each allocated with
$\ceil{\log(\ell\log^2 n + 1)}$ bits.
\item \total{} - the sum of elements inside the current frame.
\item $ind$ - the index of the most recent item, modulo $n$.
\item $T$ - a lookup table mapping sequences of length $\le \sqrt{\log n}$ to their sums.
\item $\mathcal W$ - the last $n$ elements window.
\end{itemize}
We give a pseudo code of our \qsr{} in Algorithm~\ref{alg:qsr}.
\begin{algorithm}[t]
	\caption{An \qsr{} for $\ell+1 \le 2^{\sqrt[3]{\log n}}$}\label{alg:qsr}
	\begin{algorithmic}[1]
		\Statex Initialization: ${C}\gets\bar{0},{SC}\gets\bar{0}, ind\gets 0, \total \gets 0, \mathcal W\gets\bar{0}$
		\Statex \qquad\qquad\qquad\qquad$T\gets \parentheses{\frange{\ell}^{\le \sqrt{\log n}}\to\frange{\ell\cdot\sqrt{\log n}}} \text{ lookup table}$
		\Function{Add}{element $x$}
		\State $ind\gets ind  + 1 \mod n$
		\State $\mathcal W[ind] \gets x$
		\If {$(ind \mod \sqrt{\log n}) = 0$}\Comment{End of a sub-chunk}
			\State $sum \gets T\bigg[\mathcal{W}\big[(1+ ind-\sqrt{\log n})\mod n,\ldots, ind\big]\bigg]$ \Comment{The sub-chunk's sum}
			\State $\total \gets \total + sum$			
			\State $SC[ind/\sqrt{\log n}] \gets sum$
			\If {$(ind \mod (\log n)^2) = 0$}\Comment{End of a chunk}
				\State $C[ind / (\log n)^2] \gets \total$
				\If {$(ind \mod n) = 0$} \Comment{End of a frame, reset counter}
					\State $\total \gets 0$
				\EndIf
			\EndIf
		\EndIf
		\EndFunction
		
		\Function{Query}{$i$} \Comment{For $i\le n$}
		\State $addition \gets 0$
		\If {$i \ge ind $} \Comment{If not contained in current frame}\label{line:qsr-query-ifSmallL}
			\State $addition \gets C[0]$
		\EndIf
		\State \Return \Comment{See Figure~\ref{fig:slidingRankerQuery}}\label{line:qsrQuerySmallL}\vspace*{-0.2cm}
			\begin{align*}
			T\bigg[\mathcal{W}\Big[ind - (&ind\mod \sqrt{\log n}) + 1,\ldots, ind\Big]\bigg] \\
			+ SC\bigg[\Big\lfloor ind& / \sqrt{\log n}\Big\rfloor\bigg] + C\bigg[{\Big\lfloor{ind / (\log n)^2}\Big\rfloor}\bigg] + addition\\
			- C\bigg[&\Big\lfloor{\parentheses{(ind-i)\mod n} / (\log n)^2}\Big\rfloor\bigg] - SC\bigg[{\floor{\parentheses{(ind-i)\mod n} / \sqrt{\log n}}}\bigg]\\
			&- T\bigg[\mathcal{W}\brackets{\big\lfloor{\parentheses{(ind-i)\mod n} / \sqrt{\log n}}\big\rfloor+1,\ldots,ind - i}\bigg]
			\end{align*}		\vspace*{-0.4cm}
		\EndFunction
	\end{algorithmic}
\end{algorithm}
We now formulate the properties of the algorithm; the theorem's proof is deferred to Appendix~\ref{apx:smallRSliderProof} due to lack of space.
\begin{theorem}\label{thm:smallRCorrectness}
Algorithm~\ref{alg:qsr} is an \qsr{} that uses $\mathcal B(1+o(1))$ memory~bits.
\end{theorem}
\subsection{An \qsr[{\lnrErrSymbol}] for $\lnrErrSymbol>\minDelta$}
Similarly to the way we used \lnr[\minDelta]s to construct \lnr{}s for any $\lnrErrSymbol\in\set{2,\ldots \ell\cdot n}$, we now use the exact \qsr{} for constructing an \qsr[\lnrErrSymbol].

Intuitively, we split the stream into \emph{blocks} of size $\nu$ and construct the remainder $\mathfrak r$ gradually; whenever a block ends, we compute a new $\newInputLetter_k$ value and feed it into an exact $(z,s,\minDelta)-\emph{Sliding Ranker}$ we use as a black box. When queried, we employ our exact ranker and remainder to estimate the relevant sum, similarly to our \lnr{} queries from Section~\ref{sec:lnr}. However, if we simply sum the elements using $\mathfrak r$, it will require $\Omega(\log{\ell})$ bits; this will not allow us to remain succinct if $\mu = \Omega(1)$ as the lower bound for this case is $\lnrLBSymbol=O(n)$ and is independent of $\ell$ (given that $\mu=\lnrErrSymbol/\ell$ is fixed).
To solve this, we follow~\cite{SWATPAPER}'s approach and round every arriving element, representing it using $\mathfrak b\triangleq\ceil{\logp{n/\mu}+\log\log n}$ bits. 
That is, if $x\in\frange\ell$ arrived, we consider $Round_{\nBits}(x)\triangleq2^{-\nBits}\ell\cdot\floor{\frac{x2^{\nBits}}{\ell}}$ instead.
To compensate for the rounding error, we will need blocks of size smaller than that we used in our \lnr{} construction; specifically, we set $\nu \triangleq \max\set{\floor{\mu\cdotpa{1-1/\log n}},1}$. Additionally, when $\mu<1$ the block size has to remain $1$, so we have to compensate for the rounding error by other means; this is achieved by reducing the ``sensitivity'' to $\sensitivity\triangleq \floor{\lnrErrSymbol\cdotpa{1-1/\log n}}$.\footnote{If $\sensitivity=\minDelta$, then we simply apply the exact algorithm from the previous subsection.}
The parameters for the exact ranker are then $s\triangleq \UBnBlocks$ and $z\triangleq\floor{\mu^{-1}\nu}$.
Our algorithm uses the following variables:
\begin{itemize}
	\item \ranker{} - a $(z,s,\minDelta)-\emph{Sliding Ranker}$, as described in Section~\ref{sec:qsr}.	
	\item $\mathfrak{r}$ - tracks the sum of elements that is not yet recorded in \ranker{}.
	\item $o$ - the offset within the block.

\end{itemize}
A pseudo code of our method appears in Algorithm~\ref{alg:approxSliding}.
\begin{algorithm}[]
	\caption{An \qsr[{\lnrErrSymbol}] algorithm}\label{alg:approxSliding}
	\begin{algorithmic}[1]
		\State Initialization: $\mathfrak{r} \gets 0, o \gets 0, \ranker \gets (z,s,\minDelta)-\emph{Sliding Ranker}.\text{init()}$
		\Function{\add[\text{element }\inputVariable]}{}
		\State $\blockOffset \gets (\blockOffset+1)\mod \nu$
		\State $\mathfrak{r} \gets \mathfrak{r} + Round_{\nBits}(x)$				\label{line:sum}
		\If {$\blockOffset = 0$}\label{line:end-of-block}
		\State $\newInputLetter \gets \floor{\sensitivity^{-1} \cdot\mathfrak{r}}$ \label{line:setBit}
		\State $\mathfrak{r}\gets \mathfrak{r} - \sensitivity\cdot\newInputLetter$ \label{line:reduceS}
		\State $\ranker.\mbox{\sc Add}(\newInputLetter)$
		\EndIf
		\EndFunction		
		\Function{Query}{$i$}
		\If{$i \le o$}
		\State \Return $\mathfrak{r} - \parentheses{\sensitivity - 1/2}$
		\Else
		\State $\numBits \gets \ceil{\frac{i-o}{\nu}}$
		\State $\setBits \gets \ranker.\mbox{\sc Query}\parentheses{\numBits}$
		\State $\lastBit \gets \setBits - \ranker.\mbox{\sc Query}\parentheses{\numBits-1}$
		\State $\outBits \gets \outBitsVal$
		\State \Return $\mathfrak r - \parentheses{\sensitivity - 1/2}+\sensitivity\cdot\setBits -\ell\cdot\lastBit\cdot\outBits$\label{line:est}
		
		\EndIf		
		\EndFunction
		
	\end{algorithmic}
\end{algorithm}
Next follows a memory analysis of the algorithm with a proof given in Appendix~\ref{apx:slidingQsrSpaceProof}.
\begin{lemma}\label{lem:slidingQsrSpaceProof}
Algorithm~\ref{alg:approxSliding} requires $\sFactor\cdot\floor{n/\max\set{\floor{\mu},1}}\cdot\log\big({\ceil{\mu^{-1}} + 1}\big) + O\parentheses{\log n}$~bits.
\end{lemma}
This allows us to conclude, similarly to Theorem~\ref{thm:lnrSuccinct}, that our algorithm is succinct if the error satisfies $\lnrErrSymbol=o\parentheses{\frac{\ell\cdot n}{\log n}}$. We also note that a $\floor{\log n}$ lower bound was shown in~\cite{SWATPAPER} even when only fixed sized windows (where $i \equiv n$) are considered. Thus, our algorithm always requires at most $O(\lnrLBSymbol)$, even if the allowed error is $\Omegap{\frac{\ell\cdot n}{\log n}}$.
\begin{corollary}
Let $\ell,n,\lnrErrSymbol\in\mathbb N^+$ such that $\mu\triangleq \lnrErrSymbol/\ell$ satisfies $$\parentheses{\mu=o\parentheses{\frac{n}{\log n}}} \wedge \brackets{(\mu=o(1))\vee (\mu=\omega(1)) \vee (\mu\in\mathbb N) \vee (\mu^{-1}\in\mathbb N)},$$ 
then Algorithm~\ref{alg:approxSliding} is succinct. For other parameters, it uses  $O(\lnrLBSymbol)$ space.
\end{corollary}
The following theorem, whose proof is deferred to Appendix~\ref{apx:slidingCorrectness} due to lack of space, shows the correctness of the algorithm.
\begin{theorem}\label{thm:slidingCorrecntess}
Algorithm~\ref{alg:approxSliding} is an \qsr[\lnrErrSymbol].
\end{theorem}

\section{Discussion}
\label{sec:discussion}

In this paper, we studied the properties of data structures that support \emph{approximate} rank queries for multi sets in which each element in $\orange{n}$ appears at most $\ell$ times.
We showed a lower bound for the problem and succinct constructions that require $\sFactor$ times as much memory. We then extended our approach and provided algorithms that process data streams and handle sliding window sum queries. Unlike previous work, we do not assume that the window size is fixed but rather get it at the query time. Interestingly, we show that this is doable in constant time and an additional $\sFactor$ space factor.

In the future, we would like to study structures that allow approximate select queries in $O(1)$ time. This will allow efficient approximate-percentile computation for multi sets. We note that this is already achievable with our data structure in $O(\log n)$ time using a binary search over the rank queries. We also plan to explore the possibility of creating approximate rankers with a \emph{multiplicative} error rather than additive. Finally, we wish to extend our approach to problems other than summing; e.g., computing heavy hitters for a sliding window whose size is given at the query time.
%


\newpage
{
	\tiny
\bibliographystyle{plain}
\bibliography{references}
}
\newpage
\appendix
\section{Proof of Lemma~\ref{lem:remVal}}\label{apx:remVal}
\begin{proof}
\begin{align*}
\mathfrak r &=
{{\sum_{d=1}^{n}x_d}}-\lnrErrSymbol\cdot\sum_{\jmath=1}^{s} \newInputLetter_\jmath
={{\sum_{d=1}^{n}x_d}}-\lnrErrSymbol\cdotpa{\floor{\lnrErrSymbol^{-1}\cdot{\sum_{d=n-\nu\cdot s+1}^{n}x_d}}-\sum_{\jmath=1}^{s-1} \newInputLetter_\jmath}-\lnrErrSymbol\cdot\sum_{\jmath=1}^{s-1} \newInputLetter_\jmath\\
&=\sum_{d=1}^{n}x_d-\lnrErrSymbol\cdot{\floor{\lnrErrSymbol^{-1}\cdot{\sum_{d=n-\nu\cdot \UBnBlocks+1}^{n}x_d}}}\le \lnrErrSymbol-1 + \sum_{d=1}^{n-\nu\cdot \UBnBlocks}x_d\\
&\le \lnrErrSymbol-1 + \sum_{d=1}^{\nu-1}x_d \le \lnrErrSymbol-1 + (\nu-1)\ell.
\end{align*}
If $\mu\le 1$, then $\nu=1$ and thus $\mathfrak r\le \lnrErrSymbol-1$. Otherwise, we have $\mathfrak r\le \lnrErrSymbol-1 + (\mu-1)\ell \le 2\lnrErrSymbol-\ell-1$.\qedhere
\end{proof}
\section{Proof of Theorem~\ref{thm:smallRCorrectness}} \label{apx:smallRSliderProof}
We start with analyzing the memory requirements of our algorithm.
\begin{lemma}\label{lem:slidingMem}
	Algorithm~\ref{alg:qsr} uses $\mathcal B(1+o(1))$ memory bits.
\end{lemma}
\begin{proof}
	We have $n / (\log n)^2$ chunks, each represented using $O(\log (\ell n))$ bits. Similarly, each of the $n/\sqrt{\log n}$ sub-chunk aggregates requires $O(\logp{\ell\log n})$ bits as its value is bounded by $\logp{\ell\log^2 n}$. Our window, $\mathcal W$ uses $n\logp{ \ell + 1}$ bits, while the \total{} and $ind$ variables require $O(\log n)$ bits. Thus, the overall space consumption is $\mathcal B(1+o(1))$.\qedhere
\end{proof}
We are now ready to prove the theorem.
\begin{proof}
Denote the stream by $x_1,x_2,\ldots,x_{\lastIDX}$, such that the most recent element's index is \lastIDX{}, where $m\in[n-1]$ is the offset within the current frame and $k$ frames were completed so far. 
We assume that $k\ge 1$. The case for $k=0$ follows from similar arguments.
We start with a few straight forward observations. Notice that $C[0]$ always contains the sum of the last frame that was completed; that is, $C[0]= \sum_{d=(k-1)n+1}^{k\cdot n}x_d$. Next, for any positive $j\in\brackets{n / (\log n)^2-1}$, we have that $C[j]$ contains the sum of the last $j$-indexed chunk that was completed, i.e., 
$$C[j] = \begin{cases}
\sum_{d=k\cdot n+(j-1)\cdot(\log n)^2+1}^{k\cdot n+j\cdot(\log n)^2} x_d &\mbox{if } m\ge j\cdot(\log n)^2\\
\sum_{d=(k-1)\cdot n+(j-1)\cdot(\log n)^2+1}^{(k-1)\cdot n+j\cdot(\log n)^2} x_d &\mbox{otherwise }
\end{cases}.$$
Similarly, we have that $\forall \jmath\in[2n/\log n]$:
$$
SC[\jmath] = \begin{cases}
\sum_{d=k\cdot n+(\jmath-1)\cdot\sqrt{\log n}+1}^{k\cdot n+\jmath\cdot\sqrt{\log n}} x_d &\mbox{if } m\ge \jmath\cdot\sqrt{\log n}\\
\sum_{d=(k-1)\cdot n+(\jmath-1)\cdot\sqrt{\log n}+1}^{(k-1)\cdot n+\jmath\cdot\sqrt{\log n}} x_d &\mbox{otherwise }
\end{cases}.
$$
Given a query for $i\le n$, the goal of an \qsr{} is to return the quantity $S\triangleq \sum_{d=\lastIDX-i+1}^{\lastIDX}x_d$.
First, we express the sum of elements from the beginning of the \emph{previous} frame, $S_P$, as:
$$
S_P \triangleq \sum_{d=(k-1)n+1}^{\lastIDX}x_d = \sum_{d=(k-1)n+1}^{kn}x_d + \sum_{d=kn+1}^{\floor{(\lastIDX) / (\log n)^2}}x_d + \sum_{d=\floor{(\lastIDX) / (\log n)^2}+1}^{\floor{(\lastIDX) / \sqrt{\log n}}}x_d + \sum_{d=\floor{(\lastIDX) / \sqrt{\log n}}+1}^{\lastIDX}x_d.$$
Next, since $ind = (\lastIDX\mod n)= m$, we have that 
\begin{enumerate}
\item $C[0]= \sum_{d=(k-1)n+1}^{\lastIDX}x_d$.
\item $C\brackets{\floor{ind / (\log n)^2}}= \sum_{d=kn+1}^{\floor{ind / (\log n)^2}}x_d$.
\item $SC\brackets{\floor{ind / \sqrt{\log n}}} = \sum_{d=\floor{ind / (\log n)^2}+1}^{\floor{(\lastIDX) / \sqrt{\log n}}}x_d$.
\item $T\brackets{x_{\floor{ind / \sqrt{\log n}}+1},\ldots,x_{\lastIDX}} = \sum_{d=\floor{(\lastIDX) / \sqrt{\log n}}+1}^{\lastIDX}x_d $.
\end{enumerate}   
Notice that if $i\ge m$, these are the first four summands of Line~\ref{line:qsrQuerySmallL}; if $i<m$, then we do not add $C[0]$ to the sum.
In both cases, we are left with the need to subtract the sum of elements, starting from the beginning of the relevant frame, that are not a part of the last $i$ items.
Similarly to the above, we have that the sum from the beginning of the previous frame to the $i+1$ newest item is:
$\sum_{d=(k-1)n+1}^{\lastIDX-i} x_d$.
If the last $i$ items are all contained in the current frame (i.e., $i<m$), then we have:
$$\sum_{d=(k-1)n+1}^{\lastIDX-i} x_d= \sum_{d=(k-1)n+1}^{kn}x_d + \sum_{d=kn+1}^{\floor{(\lastIDX-i) / (\log n)^2}}x_d + \sum_{d=\floor{(\lastIDX-i) / (\log n)^2}+1}^{\floor{(\lastIDX-i) / \sqrt{\log n}}}x_d + \sum_{d=\floor{(\lastIDX) / \sqrt{\log n}}+1}^{\lastIDX-i}x_d.$$

In this case, we get:
\begin{enumerate}
	\item $C[0]= \sum_{d=(k-1)n+1}^{kn}x_d$.
	\item $C\brackets{\floor{(ind-i) / (\log n)^2}}= \sum_{d=kn+1}^{\floor{\lastIDX / (\log n)^2}}x_d$.\label{itema}
	\item $SC\brackets{\floor{(ind-i) / \sqrt{\log n}}} = \sum_{d=\floor{(\lastIDX-i) / (\log n)^2}+1}^{\floor{(\lastIDX-i) / \sqrt{\log n}}}x_d$.\label{itemb}
	\item $T[x_{\floor{(ind-i) / \sqrt{\log n}}+1},\ldots,x_{ind-i}] = \sum_{d=\floor{(\lastIDX-i) / \sqrt{\log n}}+1}^{\lastIDX-i}x_d $.\label{itemc}
\end{enumerate}   
Here, we cancel the effect of $C[0]$ simply by not adding it as one of the summands (the \emph{If} condition of Line~\ref{line:qsr-query-ifSmallL}). Quantities ~\ref{itema},\ref{itemb} and \ref{itemc} are the three subtrahends of our query procedure.
Finally, if $i\ge m$ we do add the value of $C[0]$, and thus in all cases we successfully compute the sum of the last $i$ elements.
\end{proof}
\section{Proof of Lemma~\ref{lem:slidingQsrSpaceProof}}\label{apx:slidingQsrSpaceProof}
\begin{proof}
The algorithm utilizes three variables: $\ranker$ that requires $\sFactor\cdot s\logp{z+1}$, $\mathfrak{r}$ that uses $O(\mathfrak{b}\log\nu)$ bits, and $o$ is allocated with $\ceil{\log n}$ bits. Overall, the number of bits used by our construction is
\begin{align*}
&\sFactor\cdot s\logp{z+1} + O(\mathfrak{b}\log\nu) + \ceil{\log n} \\
=& \sFactor\cdot \UBnBlocks\logp{\floor{\mu^{-1}\nu+1}+1} + O({\ceil{\logp{n/\mu}+\log\log n}}\log\nu) + O\parentheses{\log n}.
\end{align*}
Since $\nu=\max\set{\floor{\mu\cdot\sNegFactor},1}$, we get the desired bound.
\end{proof}
\section{An \qsr{} for $\ell+1 > 2^{\sqrt[3]{\log n}}$}\label{apx:largeRslidingRanker}
Here, we detail the construction for the case of large $\ell$ value. We do the same splitting into frames, chunks, and sub-chunks as before. However, the large value of $\ell$ does not allow us to succinctly store the lookup table as before. Instead, we keep for each element the sum from the beginning of its sub-chunk, similarly to our solution in Theorem~\ref{thm:exactRanker}.
Our algorithm uses the following variables:
\begin{itemize}
	\item $C$ - a cyclic buffer of ${n / (\log n)^2}$ integers, each allocated with $\ceil{\log (\ell n + 1)}$ bits.
	\item $SC$ - a cyclic buffer of $n/\sqrt{\log n}$ integers, each allocated with
	$\ceil{\log(\ell\log^2 n + 1)}$ bits.
	\item \total{} - the sum of elements  in the current frame.
	\item \subTotal{} - the sum of elements in the current sub-chunk.
	\item $ind$ - the most recent element's index, modulo $n$.
	\item $\mathcal W$ - a cyclic array that contains for each item the sum from the beginning of its sub-chunk.
\end{itemize}
We give a pseudo code of our \qsr{} in Algorithm~\ref{alg:qsrLargeR}.
Next, we analyze the properties of the algorithm.
\begin{algorithm}[t]
	\caption{An \qsr{} for $\ell+1 > 2^{\sqrt[3]{\log n}}$}\label{alg:qsrLargeR}
	\begin{algorithmic}[1]
		\Statex Initialization: ${C}\gets\bar{0},{SC}\gets\bar{0}, ind\gets 0, \total \gets 0, \subTotal\gets 0, \mathcal W\gets\bar{0}$
	\Function{Add}{Element $x$}
		\State $ind\gets ind  + 1 \mod n$
		\State $\subTotal \gets \subTotal + x$
		\State $\mathcal W[ind] \gets \subTotal$ \label{line:wGetsSubTotal}
		\If {$(ind \mod \sqrt{\log n})) = 0$}\Comment{End of a sub-chunk}
			\State $\total \gets \total + \subTotal$			
			\State $SC[ind/\sqrt{\log n}] \gets \subTotal$
			\State $\subTotal \gets 0$			
			\If {$(ind \mod (\log n)^2) = 0$}\Comment{End of a chunk}
				\State $C[ind / (\log n)^2] \gets \total$
				\If {$(ind \mod n) = 0$} \Comment{End of a frame, reset counter}
					\State $\total \gets 0$
				\EndIf
			\EndIf
		\EndIf
	\EndFunction
		
		\Function{Query}{$i$} \Comment{For $i\le n$}
		\State $addition \gets 0$
		\If {$i \ge ind $} \Comment{If not contained in current frame}\label{line:qsr-query-if}
		\State $addition \gets C[0]$
		\EndIf
		\State \Return \Comment{See Figure~\ref{fig:slidingRankerQuery}}\label{line:qsrQuery}\vspace*{-0.2cm}
		\begin{align*}
		\subTotal& + SC\bigg[\Big\lfloor ind / \sqrt{\log n}\Big\rfloor\bigg] + C\bigg[{\Big\lfloor{ind / (\log n)^2}\Big\rfloor}\bigg] + addition\\
		- C\bigg[&\Big\lfloor{\parentheses{(ind-i)\mod n} / (\log n)^2}\Big\rfloor\bigg] - SC\bigg[{\floor{\parentheses{(ind-i)\mod n} / \sqrt{\log n}}}\bigg] - \mathcal{W}\brackets{ind - i}
		\end{align*}		\vspace*{-0.4cm}
		\EndFunction
		\end{algorithmic}
\end{algorithm}

\begin{theorem}
	Algorithm~\ref{alg:qsrLargeR} uses $\lnrLBSymbol(1+o(1))$ memory bits for $\ell+1 > 2^{\sqrt[3]{\log n}}$.
\end{theorem}
\begin{proof}
	Similarly to the analysis in Lemma~\ref{lem:slidingMem}, the algorithm uses $o(\lnrLBSymbol)$ bits for keeping the chunk and sub-chunk aggregates. Here, we replaced the lookup table \emph{and} the array of window elements by an array that stores the within-sub-chunk cumulative sum for each element. That is, each entry in the $n$-sized array stores a number in $\frange{\ell\cdot\sqrt{\log n}}$ and thus the array requires $n\cdot\logp{\ell\cdot\sqrt{\log n}+1}\le 
	n\log(\ell+1)\parentheses{1+\frac{\log{\log n}}{\logp{\ell+1}}}=\sFactor\lnrLBSymbol$ 
	bits~overall.\qedhere
\end{proof}
\begin{theorem}
	Algorithm~\ref{alg:qsrLargeR} is an \qsr{}.
\end{theorem}
\begin{proof}
Observe that the query procedure of our algorithm is equivalent to that of Algorithm~\ref{alg:qsr}, except for the part where it uses the lookup table.
We now use the $\subTotal$ variable to compute the sum of the current sub-chunk instead of looking it up in the table. The sum of elements that preceded the last $i$ items in $(ind - i)$'s sub-chunk is then retrieved from $\mathcal{W}\brackets{ind - i}$. As we simply track the sum of items prior to that index in $\subTotal$ and then store it in $\mathcal W$ (see line~\ref{line:wGetsSubTotal}), we get its value immediately. Thus, our estimation procedure is equivalent to that of Algorithm~\ref{alg:qsr} and using Theorem~\ref{thm:smallRCorrectness} we establish our correctness.
\end{proof}
\newcommand{\lastT}{h}
\section{Proof of Theorem~\ref{thm:slidingCorrecntess}}\label{apx:slidingCorrectness}
\begin{proof}
For the proof, we define a few quantities that we also use in our query procedure $\numBits \triangleq \ceil{\frac{i-o}{\nu}}, \setBits \triangleq \ranker.\mbox{\sc Query}\parentheses{\numBits}, \lastBit \triangleq \setBits - \ranker.\mbox{\sc Query}\parentheses{\numBits-1}$ and $\outBits \triangleq \outBitsVal$ as in our {\sc Query} function; see Figure~\ref{fig:slidingRankerQueryProof} for illustration.
We assume that the index of the most recent element is 
$$\lastT\triangleq \blockOffset+\outBits,$$ such that $o\in[\nu-1]$ is the offset within the current block and that $x_1$ is the first element in the newest block of $\lastBit$.
By the correctness of the $\ranker$ Sliding Ranker, and as illustrated in Figure~\ref{fig:slidingRankerQueryProof}, we have that \setBits{} is the sum of the last \numBits{} added to $\ranker$, that \lastBit{} is the value of the element that represents the last block that overlaps with the queried window. Also notice that \outBits{} is the number of elements in that block that are not a part of the  window.
\begin{figure}[]
	\centering
	\includegraphics[width=\linewidth]{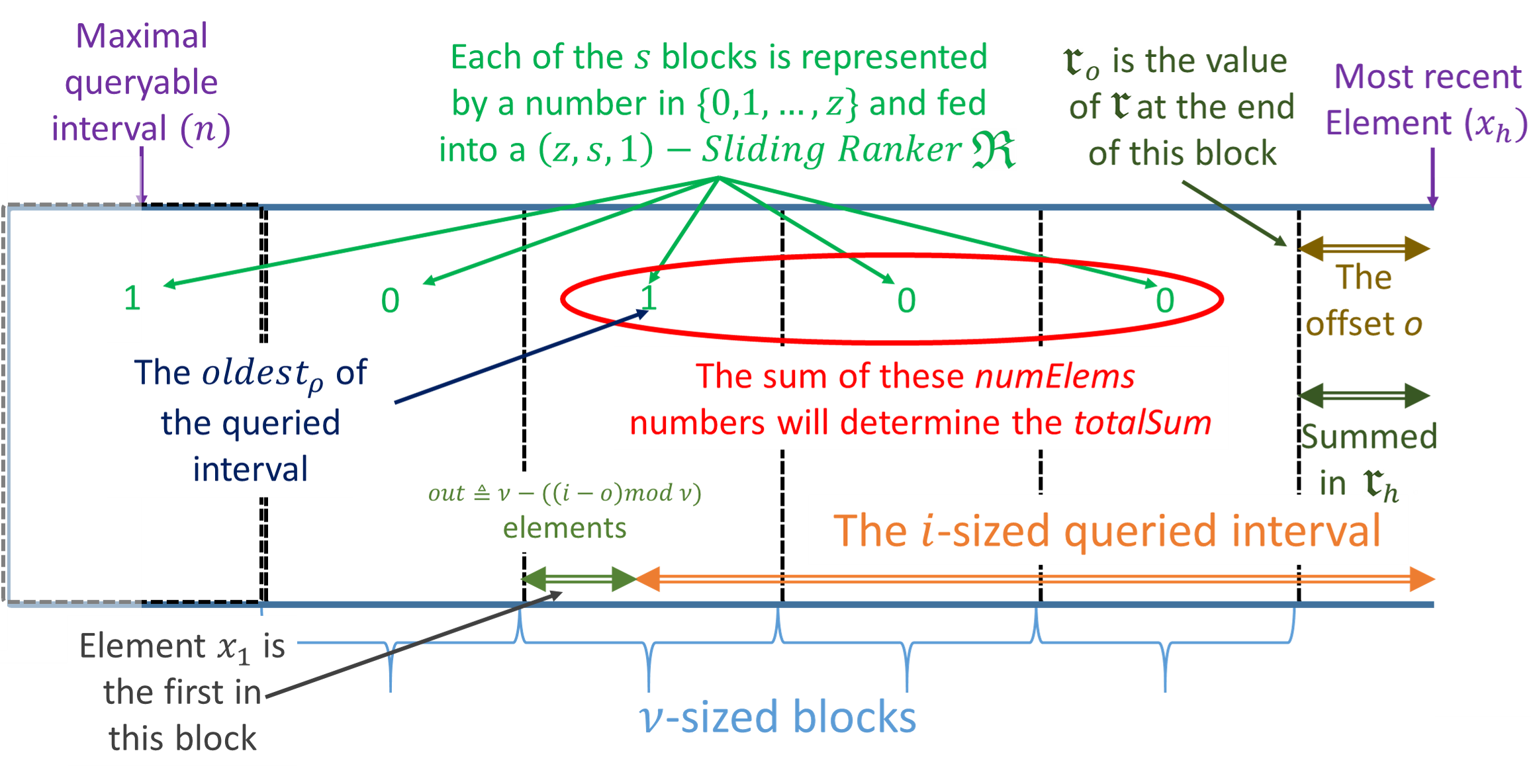}
	\caption{Theorem~\ref{thm:slidingCorrecntess} proof's setting, with all relevant quantities that Algorithm~\ref{alg:approxSliding} uses illustrated.}
	\label{fig:slidingRankerQueryProof}
\end{figure}
For any $t\in \mathbb N$, we denote by $\mathfrak{r_t}$ the value of $\mathfrak{r}$ \emph{after} the $t^{th}$ item was added; e.g., $\mathfrak{r_{\lastT}}$ is the value of $\mathfrak{r}$ at the time of the query and $\mathfrak{r_0}$ is its value after the last block has ended. For other variables, we consider their value at the query time.

When a block ends, we effectively perform $\mathfrak{r}\gets \mathfrak{r}\mod \sensitivity$ (lines~\ref{line:setBit} and \ref{line:reduceS}) and thus:
\begin{align}
0 \le \mathfrak{r_0} \le \sensitivity - 1.\label{eq:s0}
\end{align}
Our goal is to estimate the quantity 
\begin{align}
S_i \triangleq \sum_{d=\lastT-i+1}^{\lastT} x_d = \sum_{d=\outBits+1}^{\lastT} x_d.\label{eq0}
\end{align}
Recall that our estimation (Line~\ref{line:est}) is:
\begin{multline}
\widehat{S_i}\triangleq\mathfrak{r_{\lastT}} - \parentheses{\sensitivity - 1/2}+\sensitivity\cdot\setBits -\ell\cdot\lastBit\cdot\outBits\\
= \mathfrak{r_{\outBits}} + \sum_{d=\outBits+1}^{\lastT} Round_{\nBits}(x_d) - \parentheses{\sensitivity - 1/2}+\sensitivity\cdot\setBits -\ell\cdot\lastBit\cdot\outBits,
\label{eq1}
\end{multline}
where the last equality follows from the fact that within a block we simply sum the rounded values (Line~\ref{line:sum}).
Next, observe that we sum the rounded values in each block and that if $\mathfrak{r}$ is decreased by $k\cdot \sensitivity$ (for some $k\in\mathbb N$) in Line~\ref{line:reduceS}, then we set one of the last $\numBits$ elements added to $\ranker$ to $k$. This means that:
\begin{align}
\mathfrak{r_0} + \sum_{d=1}^{\outBits} Round_{\nBits}(x_d) 
= \mathfrak{r_{\outBits}} + \sensitivity\cdot\ranker.\mbox{\sc Query}\parentheses{\numBits} = \mathfrak{r_{\outBits}} + \sensitivity\cdot\setBits.\label{eq2}
\end{align}
Plugging~\eqref{eq2} into~\eqref{eq1} gives us
\begin{multline}
\widehat{S_i}= \mathfrak{r_0} + \sum_{d=1}^{\outBits} Round_{\nBits}(x_d)  + \sum_{d=\outBits+1}^{\lastT} Round_{\nBits}(x_d) - \parentheses{\sensitivity - 1/2} -\ell\cdot\lastBit\cdot\outBits.
\label{eq3}
\end{multline}
Joining~\eqref{eq3} with \eqref{eq0}, we can express the algorithm's error as:
\begin{multline}
\widehat{S_i} - S_i = \mathfrak{r_0} + \sum_{d=1}^{\outBits} Round_{\nBits}(x_d)  + \sum_{d=\outBits+1}^{\lastT} \biggParentheses{Round_{\nBits}(x_d) - x_d} - \parentheses{\sensitivity - 1/2} -\ell\cdot\lastBit\cdot\outBits\\
= \mathfrak{r_0} + \sum_{d=1}^{\outBits} Round_{\nBits}(x_d)  + \xi - \parentheses{\sensitivity - 1/2} -\ell\cdot\lastBit\cdot\outBits
,\label{eq4}
\end{multline}
where $\xi$ is the rounding error which is defined as
\begin{align*}
\xi \triangleq \sum_{d=\outBits+1}^{\lastT} \biggParentheses{Round_{\nBits}(x_d) - x_d}.
\end{align*}
Since each rounding of an integer $x\in\frange{\ell}$ has an error of at most $\frac{\ell}{2^{\nBits}}$, and as
we round $i\le n$ elements, we have that the rounding error satisfies
\begin{align}
0 \ge \xi \ge 0 -\frac{\ell\cdot n}{2^{\nBits}} \ge -\lnrErrSymbol/\log n
,\label{eq5}
\end{align}
where the last inequality is immediate from our choice of the number of bits that is $\nBits\triangleq\ceil{\logp{n/\mu}+\log\log n}$.
We now split to cases based on the value of $\mu$. As in the \lnr{} case, we start with the simpler $\mu< 2\cdotpa{1-1/\log n}$ case, in which $\nu=1$ (and consequently, $out\equiv 0$). 
This allows us to write the algorithm's error of~\eqref{eq4} as
\begin{align}
\widehat{S_i} - S_i = \mathfrak{r_0}  + \xi - \parentheses{\sensitivity - 1/2}.
\end{align}
We now use \eqref{eq:s0},\eqref{eq5} and the definition of \sensitivity{} to obtain:
\begin{align*}
\widehat{S_i} - S_i = \mathfrak{r_0}  + \xi - \parentheses{\sensitivity - 1/2} \le -1/2.
\end{align*}
Similarly, we can bound it from below:
\begin{align*}
\widehat{S_i} - S_i = \mathfrak{r_0}  + \xi - \parentheses{\sensitivity - 1/2} \ge \xi - \parentheses{\sensitivity - 1/2} \ge -\lnrErrSymbol + 1/2.
\end{align*}
We established that if $\nu=1$ we obtain the desired approximation.
Henceforth, we focus on the case where $\mu\ge 2\cdotpa{1-1/\log n}$, and thus $\nu=\floor{\mu\cdotpa{1-1/\log n}}$ and $\lastBit\in\set{0,1}$.
We now consider two cases, based on the value of \lastBit.
\begin{enumerate}
	\item {\textbf{$\bm{\lastBit=1}$ case.}\\}
	In this case, we know that after the processing of element $x_\nu$ the value of $\mathfrak{r}$ was at least $\sensitivity$ (Line~\ref{line:setBit}). This implies that $\mathfrak{r_0} + \sum_{d=1}^{\nu} Round_{\nBits}(x_d) \ge \sensitivity$ and equivalently
	\begin{align*}
	\mathfrak{r_0} + \sum_{d=1}^{\outBits} Round_{\nBits}(x_d) \ge \sensitivity - \sum_{d=\outBits+1}^{\nu} Round_{\nBits}(x_d).
	\end{align*}
	Substituting this in~\eqref{eq4}, and applying~\eqref{eq5}, we get that:
	\begin{align*}
	\widehat{S_i} - S_i &= \mathfrak{r_0} + \sum_{d=1}^{\outBits} Round_{\nBits}(x_d)  + \xi - \parentheses{\sensitivity - 1/2} -\ell\cdot\outBits\\
	&\ge \sensitivity - \sum_{d=\outBits+1}^{\nu} Round_{\nBits}(x_d)  + \xi - \parentheses{\sensitivity - 1/2} -\ell\cdot\outBits\\
	&\ge - \parentheses{\sum_{d=\outBits+1}^{\nu} \ell}  + \xi + 1/2 -\ell\cdot\outBits\\
	&\ge -\lnrErrSymbol/\log n -\ell\floor{\mu\cdotpa{1-1/\log n}}+1/2 \ge - \lnrErrSymbol + 1/2.
	\end{align*}
	In order to bound the error from above we use \eqref{eq:s0} and \eqref{eq5}:
	\begin{align*}
	\widehat{S_i} - S_i &= \mathfrak{r_0} + \sum_{d=1}^{\outBits} Round_{\nBits}(x_d)  + \xi - \parentheses{\sensitivity - 1/2} -\ell\cdot\outBits\\
	&\le \sensitivity - 1 + \ell\cdot \outBits  - \parentheses{\sensitivity - 1/2} -\ell\cdot\outBits \le -1/2.
	\end{align*}
	\item {\textbf{$\bm{\lastBit=0}$ case.}\\}	
	Here, since the value of \lastBit{} was is $0$, we have that $\mathfrak{r_0} + \sum_{d=1}^{\nu} Round_{\nBits}(x_d) < \sensitivity$ and thus
	\begin{align*}
	\mathfrak{r_0} + \sum_{d=1}^{\outBits} Round_{\nBits}(x_d) \le \sensitivity - \sum_{d=\outBits+1}^{\nu} Round_{\nBits}(x_d) - 1.
	\end{align*}
	We use this for the error expression of~\eqref{eq4} to get:
	\begin{align*}
	\widehat{S_i} - S_i &= \mathfrak{r_0} + \sum_{d=1}^{\outBits} Round_{\nBits}(x_d)  + \xi - \parentheses{\sensitivity - 1/2}\\
	&\le \sensitivity - \sum_{d=\outBits+1}^{\nu} Round_{\nBits}(x_d) - 1  + \xi - \parentheses{\sensitivity - 1/2} \le -1/2\\
	\end{align*}
	We now use \eqref{eq:s0}, \eqref{eq5}, and the fact that $\outBits\le \nu$ to bound the error from below as follows:
	\begin{align*}
	\widehat{S_i} - S_i &= \mathfrak{r_0} + \sum_{d=1}^{\outBits} Round_{\nBits}(x_d)  + \xi - \parentheses{\sensitivity - 1/2}\\
	&\ge \xi - \parentheses{\sensitivity - 1/2} \ge -\lnrErrSymbol + 1/2.
	\end{align*}
\end{enumerate}
Finally, we need to cover the case of $i\le o$. In this case, we return $\mathfrak r-\parentheses{\sensitivity - 1/2}$ as the estimate. This directly follows from~\eqref{eq:s0} and the fact that within a block we simply sum the rounded values (Line~\ref{line:sum}).
We established that in all cases $-\lnrErrSymbol < \widehat{S_i} - S_i < 0$, thereby proving the theorem.
\end{proof}

\end{document}